%% file: main.tex
\documentclass[10pt,nocopyrightspace]{sigplanconf}

\usepackage{amsmath}
\usepackage{amssymb}
\usepackage{ntheorem}
\usepackage[english]{babel}
\usepackage{xcolor}
\usepackage{graphicx}
\usepackage{pgfplots}
\usepackage{textcomp}
\usepackage{listings}
\usepackage{footnote}
\usepackage{proof}
\usepackage{subcaption}
\usepackage{multirow}
\usepackage{url}
\usepackage{hyperref}
\usepackage{breakurl}
\usepackage{float}
\usepackage{tikz}
\usetikzlibrary{arrows,%
                petri,%
                topaths}%
\usepackage{tkz-berge}
\usetikzlibrary{calc}

\usepackage{tabularx}

\floatstyle{plain}
\newfloat{Listing}{htb}{prg}%[section]

\usepackage{caption}

\newcommand{\figref}[1]{Figure~\ref{#1}}
\newcommand{\tabref}[1]{Table~\ref{#1}}
\newcommand{\secref}[1]{Section~\ref{#1}}

\newcommand{\java}[1]{{\upshape\lstinline[language={Java},basicstyle=\normalsize]{#1}}}
\newcommand{\sjava}[1]{{\upshape\lstinline[language={Java},basicstyle=\small]{#1}}}

\newcommand{\rulename}[1]{\ensuremath{\mbox{\textsc{#1}}}}
\newcommand{\syntax}[1]{\ensuremath{\mathtt{#1}}}
\newcommand{\function}[1]{\mbox{\normalfont\sffamily#1}}
\newcommand{\owner}{\function{root}}
\newcommand{\acyclic}{\function{acyclic}}
\newcommand{\fresh}{\function{fresh}}

\newtheorem{definition}{Definition}
\newtheorem{example}{Example}
\newtheorem{theorem}{Theorem}
\newtheorem*{proof}{Proof}
\newtheorem{lemma}{Lemma}

\begin{document}

\setlength{\pdfpageheight}{\paperheight}
\setlength{\pdfpagewidth}{\paperwidth}

%\conferenceinfo{OOPSLA 2014}{October 20--24, 2014, Portland, Oregon, United States of America} 
%\copyrightyear{2014} 
%\copyrightdata{978-1-nnnn-nnnn-n/yy/mm} 
%\doi{nnnnnnn.nnnnnnn}

% Uncomment one of the following two, if you are not going for the 
% traditional copyright transfer agreement.

%\exclusivelicense                % ACM gets exclusive license to publish, 
                                  % you retain copyright

%\permissiontopublish             % ACM gets nonexclusive license to publish
                                  % (paid open-access papers, 
                                  % short abstracts)

%\titlebanner{banner above paper title}        % These are ignored unless
%\preprintfooter{short description of paper}   % 'preprint' option specified.

\title{Dynamic Checking of Safe Concurrent Memory Access using Shared Ownership}
%\subtitle{Subtitle Text, if any}

\authorinfo{Mischael Schill, Sebastian Nanz, and Bertrand Meyer}
           {ETH Z\"urich, Switzerland}
           {firstname.lastname@inf.ethz.ch}
\maketitle

\date{17 April 2014}
\begin{abstract}
In shared-memory concurrent programming, shared resources can be protected using synchronization mechanisms such as monitors or channels. The connection between these mechanisms and the resources they protect is, however, only given implicitly; this makes it difficult both for programmers to apply the mechanisms correctly and for compilers to check that resources are properly protected. This paper presents a mechanism to automatically check that shared memory is accessed properly, using a methodology called shared ownership. In contrast to traditional ownership, shared ownership offers more flexibility by permitting multiple owners of a resource. %This enables the safe use of readers-writer locks, in addition to other synchronization mechanisms. 
On the basis of this methodology, we define an abstract model of resource access that provides operations to manage data dependencies, as well as sharing and transfer of access privileges. The model is rigorously defined using a formal semantics, and shown to be free from data races. This property can be used to detect unsafe memory accesses when simulating the model together with the execution of a program. The expressiveness and efficiency of the approach is demonstrated on a variety of programs using common synchronization mechanisms.
\end{abstract}

\section{Introduction}
\label{sec:introduction}
Protecting the access to shared resources is one of the fundamental challenges in concurrent programming. Without taking the appropriate precautions, concurrent access of multiple processes to the same resource can result in atomicity violations such as data races, which are difficult to detect as they often only manifest in a fraction of all the possible interleavings of a concurrent program. To avoid this issue, programmers must protect shared resources using synchronization mechanisms, such as monitors, channels or semaphores. Unfortunately, using such synchronization mechanisms is a delicate task. By mistake, programmers may take too few or too many locks, locks in the wrong order, or the wrong locks altogether. The primary reason for these difficulties is the missing connection between locks and the shared resources they are supposed to protect. To be more precise, a lock does not protect state but rather sections of code, which are called critical sections; the shared resource is protected indirectly by the unwritten contract that it may only be accessed from within a critical section protected by a specific lock.

An approach to prove that shared resources are protected is the ownership methodology~\cite{Clarke:1998:OTF:286936.286947}, which has originally been designed for static reasoning on sequential programs but can be extended to handle concurrency~\cite{Boyapati:2002:OTS:582419.582440,jacobs2005safe}. The central idea of ownership is that each object has an \textit{owner}, which is either void (the object is free), an object (the object is part of an aggregate), or a process (the object is owned by a process).
%
% \begin{figure}
%   \centering
%
% \begin{tikzpicture}[scale=0.75,transform shape]
%   \Vertex[x=0,y=2,L=$\pi$]{p1}
%   \Vertex[x=2,y=2,L=$\rho_1$]{r1}
%   \Vertex[x=4,y=2,L=$\rho_3$]{r3}
%   \Vertex[x=4,y=1,L=$\rho_4$]{r4}
%   \Vertex[x=2,y=0,L=$\rho_2$]{r2}
%   \Vertex[x=4,y=0,L=$\rho_5$]{r5}
%
%   \tikzstyle{LabelStyle}=[fill=white,sloped]
% //  \tikzstyle{EdgeStyle}=[bend left]
%   \tikzstyle{EdgeStyle}=[post]
%   \Edge[](p1)(r1)
%   \Edge[](r1)(r3)
%   \Edge[](r1)(r4)
%
%   \Edge[](r2)(r5)
% \end{tikzpicture}
%
%   \caption{Two simple ownership trees}
%   \label{fig:ownership-trees}
% \end{figure}
%
% \begin{example}
% Ownership relations are essentially trees of objects. Figure \figref{fig:ownership-trees} shows two very simple trees: a thread $\pi$ who owns the aggregate object $\rho_1$, which is composed of the two parts $\rho_3$ and $\rho_4$, and the free aggregate object $\rho_2$. The source of an edge denotes the owner, the sink denotes the owned object.
% \end{example}
%
In order to access an object, a process must acquire ownership of the object. A free object can be acquired by acquiring its lock; upon successful locking, the owner is set to the process that holds the lock. An object that is part of an aggregate cannot be acquired directly, but rather only by acquiring the free object at root of the ownership tree created through the ownership relations.
%If an aggregate object is owned by the process, its parts can be made writable by using the \textit{unpack} operation, accessing the state without modification is possible without unpacking. This operation also suspends the invariant of the aggregate object until the \textit{pack} operation is used, using the Boogie methodology\cite{barnett2004verification}.

There are some limitations with this approach: it concentrates on using locks for handling exclusive access; it does not incorporate temporary shared reading, e.g.\ through readers-writer locks; and it statically specifies ownership relations through annotations on attributes. These limitations stem from the application of ownership mainly for modular static verification of programs. Only recently there has been some work~\cite{Martin:2010} on using ownership principles as a methodology for structuring concurrent programs.% and then using simple dynamic methods for checking adherence.

This paper proposes a variant of the ownership methodology called \textit{shared ownership} that incorporates shared reading and can therefore support more synchronization techniques. Programs adhering to this methodology do not exhibit data races. As an application of shared ownership, we present a dynamic method to automatically check adherence of Java programs to the methodology.

In shared ownership, an object has always at least one owner, which may be either another object or a process. %The notion of process encompasses both processes and synchronization measures such as locks or channels. Although objects are often considered as the smallest element of state, they can be broken down further, for example in the case of arrays~\cite{schill2013parallel}. Therefore we use the notion of a \emph{resource} as the smallest element of state. Since shared ownership allows multiple owners, it operates on directed graphs rather than trees. %These ownership graphs are detailed in section \secref{sec:rcg}. We further define a set of operations that operate on \textit{ownership graph}s through an operational semantics. These operations omit synchronization mechanisms, since this is done at the level of the programming language or framework. However, they guarantee freedom from race conditions as long as their premises are met.
The ownership relations create a directed graph, rather than a tree, since multiple owners are allowed. Ownership graphs can be used to ensure mutual exclusion on write operations and to define a set of readers for read operations: a process may modify the state of the resources that it exclusively owns, i.e.\ all resources where the process is the sole root of the ownership subgraph that includes the resource; a process has read access to all resources it owns, i.e.\ all resources where the process is a root of the ownership subgraph that includes the resource. %Tracking information about shared ownership can therefore be used for debugging and automatic checking of safe resource access. %application of concurrency patterns.

To track information about shared ownership, we define 
%We present an abstract methodology, 
Shared Ownership Model (SOM), which provides an abstract model of resource access and describes operations on ownership graphs, such as transfer of ownership between processes. We rigorously specify the methodology using operational semantics, and show that it ensures freedom from data races.
%
%To demonstrate the methodology's practicality, 
%
%We integrate SOM into the programming language Java by describing the correspondence between program instructions and SOM statements and expressing common synchronization constructs in SOM. 
An integration of SOM into the programming language Java allows us to use SOM to dynamically check Java programs for data races: the ownership dependencies arising in the program are simulated using SOM; if the simulation fails, an ownership violation has been found, which may manifest itself in a data race. 
We evaluate the performance of our prototype, confirming that the approach is efficiently usable. Furthermore, since the ownership graph does not affect the functional aspects of a program, it can be removed by the compiler to optimize the program for production environments, where the information provided by the ownership graph is no longer needed.

The remainder of the paper is structured as follows. \secref{sec:rcl} specifies SOM using an operational semantics, and provides the proof of data-race freedom. \secref{sec:integration} describes the integration of the methodology into Java, expresses various synchronization mechanisms in terms of the abstract model, and demonstrates how the approach can be used to dynamically check for data races. \secref{sec:performanceanalysis} is dedicated to the performance analysis of using the approach as a debugging tool. \secref{sec:related-work} presents related work and \secref{sec:conclusion} concludes the paper with an outlook on future work.

\section{Shared Ownership}
\label{sec:rcl}

This section presents Shared Ownership Model, an abstract model of resource access. It uses the formalism of ownership graphs, which specify the access rights of processes to shared resources. We show that the model is race-free.%free from data races.

\subsection{Ownership graphs}
\label{sec:rcg}

We define an entity $e$ to be either a process $\pi$ or a (shared) resource $\rho$. Ownership graphs express ownership relationships between entities.
\begin{definition}[Ownership graph]\label{def:rcg}
An \emph{ownership graph} is a directed graph over a set of entities such that the following conditions hold:
\begin{itemize}
\item[(P)] No process node has an incoming edge.
\item[(R)] All resource nodes have at least one incoming edge.
\item[(A)] The graph is acyclic.
\end{itemize}
\end{definition}
If an entity $e$ is a predecessor of a resource $\rho$ in an ownership graph, $\rho$ is said to be \emph{owned} by $e$. If a resource $\rho$ is reachable from a process $\pi$, we say that $\pi$ is a \emph{root owner}, or simply a \emph{root}, of $\rho$. Processes have no owner and are thereby always roots. Resources on the other hand are always owned, therefore they are never roots. A circular ownership relation would mean that the involved entities mutually own each other, therefore we define ownership graphs to be acyclic.

\begin{example}
  \figref{fig:rcg-example} shows an ownership graph over two processes $\pi_1$ and $\pi_2$, and four resources $\rho_1, \dots, \rho_4$. From the graph, we can see that $\rho_3$ is owned by $\rho_1$, and has process $\pi_1$ as root. Resource $\rho_4$ is owned by both $\rho_1$ and $\rho_2$, and has two roots $\pi_1$ and $\pi_2$. 
\end{example}

\begin{figure}[htb]
\vspace{-2ex}
  \centering

\begin{tikzpicture}[scale=0.75,transform shape]
  \Vertex[x=0,y=2,L=$\pi_1$]{p1}
  \Vertex[x=0,y=1,L=$\pi_2$]{p2}
  \Vertex[x=2,y=2,L=$\rho_1$]{r1}
  \Vertex[x=4,y=2,L=$\rho_3$]{r2}
  \Vertex[x=4,y=1,L=$\rho_4$]{r3}
  \Vertex[x=2,y=1,L=$\rho_2$]{r4}

  \tikzstyle{LabelStyle}=[fill=white,sloped]
//  \tikzstyle{EdgeStyle}=[bend left]
  \tikzstyle{EdgeStyle}=[post]
  \Edge[](p1)(r1)
  \Edge[](r1)(r2)
  \Edge[](r1)(r3)

  \Edge[](p2)(r4)
  \Edge[](r4)(r3)
\end{tikzpicture}
\vspace{-1ex}
  \caption{A simple ownership graph}
  \label{fig:rcg-example}
\vspace{-1ex}
\end{figure}
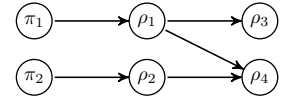

Within our methodology, ownership graphs will be used as a means to specify the read and write access rights of different entities. If a resource $\rho$ has a single root $\pi$, then $\pi$ is said to have \emph{write access} to $\rho$ and is allowed to both read and write $\rho$. If a resource $\rho$ has multiple roots, they are said to have \emph{read access} to $\rho$.

\subsection{Shared Ownership Model}
\label{sec:semantics}

We introduce \emph{Shared Ownership Model (SOM)}, which specifies operations on ownership graph.

\paragraph{Syntax.} The syntax of SOM is shown in \figref{fig:syntax}. It consists of configurations $C$, processes $P$, entities $e$, and statements $s$. 
A configuration is a pair of a process $P$ and an ownership graph $\sigma$. Programs consist of parallel processes of the form $\pi : s$, where $\pi$ is the identifier of the process and $s$ is a statement. Entities $e$ are differentiated into processes, described by their identifier $\pi$, and resources $\rho$.

\begin{figure}[htb]
\vspace{-2ex}
{\small
\[
\begin{array}{lll@{\quad}l}
C &::= &\langle P, \sigma \rangle& \mbox{Configuration}\\[.5ex]
P &::= &\pi : s& \mbox{Process}\\
	&| & P | P & \mbox{Parallel composition}\\[.5ex]
e &::= &\pi& \mbox{Process identifier}\\
	&| &\rho& \mbox{Resource}\\[.5ex]
s &::= &\rho.\syntax{read}& \mbox{Read from } \rho\\
	&| &\rho.\syntax{write}& \mbox{Write to } \rho\\
	&| &\rho.\syntax{pass}(e, e')& \mbox{Pass ownership over } \rho \\
	& & & \mbox{from } e \mbox{ to } e'\\
	&| &\rho.\syntax{share}(e)& \mbox{Share ownership over } \rho \mbox{ with } e\\
	&| &\rho.\syntax{release}(e)& \mbox{Release ownership over } \rho \mbox{ by } e\\
	&| &\pi := \syntax{spawn}(s)& \mbox{Spawn a new process } \pi\\
	&| &\rho := e.\syntax{allocate}& \mbox{Allocate a new resource } \rho \\
	& & & \mbox{owned by } e
%	&| &\syntax{skip} & \mbox{No behavior}\\
%	&| &s_1; s_2 & \mbox{Sequential composition}
\end{array}
\]
}
\vspace{-3.5ex}
\caption{The syntax of SOM}
\label{fig:syntax}
%\vspace{-2ex}
\end{figure}

Statements model the read/write access to a resource $\rho$ using the statements $\rho.\syntax{read}$ and $\rho.\syntax{write}$. Three statements allow the executing process $\pi$ to modify an ownership graph: if $\pi$ has write access over $\rho$, it can use the statement $\rho.\syntax{pass}(e, e')$ to pass ownership over resource $\rho$ from entity $e$ to entity $e'$; if $\pi$ has read access of $\rho$, it can use the statement $\rho.\syntax{share}(e)$ to share ownership over resource $\rho$ with entity $e$; if $\pi$ has write access over $e$, it can use the statement $\rho.\syntax{release}(e)$ to  release ownership by $e$ over $\rho$. A new process $\pi' : s$ can be spawned using $\pi' := \syntax{spawn}(s)$. A new resource $\rho$ owned by $e$ can be allocated using $\rho := e.\syntax{allocate}$. 
%The statement $\syntax{skip}$ models no behavior, and statements can be composed sequentially $s_1; s_2$.

\paragraph{Semantics.} The definition of the operational semantics requires three auxiliary functions: $\owner_\sigma(e)$ retrieves the root owners of entity $e$ from the ownership graph $\sigma$; $\acyclic(\sigma)$ determines whether $\sigma$ is acyclic; finally, $\fresh_\sigma(e)$ denotes that the entity $e$ is not present in $\sigma$.

\figref{fig:inference} shows the inference rules of SOM. Rules $\rulename{par}_1$ and $\rulename{par}_2$ are rules for parallel composition of general processes $P$. All other rules specify the execution of statements for a currently executing process $\pi$.
% Rules $\rulename{serial}$ and $\rulename{skip}$ straightforwardly express sequential composition and no behavior.
%
Inference rule $\rulename{read}$ allows $\pi$ to read a resource $\rho$, if the current ownership graph $\sigma$ specifies $\pi$ as a root of $\rho$; rule $\rulename{write}$ specifies that $\pi$ can write $\rho$ if it is the only root of $\rho$. 
Rule $\rulename{pass}$ allows $\pi$ to pass ownership over $\rho$ from $e$ to $e'$ if it is the only root of $\rho$, $e$ directly owns $\rho$, and the passing of ownership does not introduce cycles into the ownership graph. 

\newcommand{\infspacing}{1.5ex}
\begin{figure*}
\centering
{\small
\begin{gather*}
	\infer[\rulename{par}_1]
		{\langle P_1 | P_2,\sigma \rangle 
			\rightarrow 
		\langle P_1' | P_2,\sigma' \rangle}
		{\langle P_1,\sigma \rangle 
			\rightarrow 
		\langle P_1',\sigma' \rangle} \quad
	\infer[\rulename{par}_2]
		{\langle P_1 | P_2,\sigma \rangle 
			\rightarrow 
		\langle P_1 | P_2',\sigma' \rangle}
		{\langle P_2,\sigma \rangle 
			\rightarrow 
		\langle P_2',\sigma' \rangle} \\[\infspacing]
%	\infer[\rulename{serial}]
%		{\langle \pi : s_1 ; s_2,\sigma \rangle 
%			\rightarrow 
%		\langle \pi : s_1' ; s_2,\sigma' \rangle}
%		{\langle s_1,\sigma \rangle 
%			\rightarrow 
%		\langle s_1',\sigma' \rangle} \quad	
%	\infer[\rulename{skip}] 
%		{\langle \pi : \syntax{skip} ; s_2,\sigma \rangle 
%			\rightarrow 
%		\langle \pi : s_2,\sigma \rangle}
%		{} \\\\	
	\infer[\rulename{read}]
		{\langle \pi : \rho.\syntax{read} ; s,\sigma \rangle 
			\rightarrow 
		\langle \pi : s,\sigma \rangle}
		{\pi \in \owner_\sigma(\rho)} \quad	
	\infer[\rulename{write}]
		{\langle \pi : \rho.\syntax{write} ; s,\sigma \rangle 
			\rightarrow 
		\langle \pi : s,\sigma \rangle}
		{\{\pi\} = \owner_\sigma(\rho)} \\[\infspacing]	
	\infer[\rulename{pass}]
		{\langle \pi : \rho.\syntax{pass}(e, e') ; s,\sigma \rangle 
			\rightarrow 
		\langle \pi : s,\sigma'  \rangle} 
		{ \{\pi\} = \owner_\sigma(\rho) \land \{e \rightarrow \rho\} \subseteq \sigma \land \sigma' = ((\sigma \setminus \{e \rightarrow \rho\})\cup\{e' \rightarrow \rho\}) \land \acyclic(\sigma') } \\[\infspacing]
	\infer[\rulename{share}]
		{\langle \pi : \rho.\syntax{share}(e) ; s,\sigma \rangle 
			\rightarrow 
		\langle \pi : s,\sigma' \rangle} 
		{ \pi \in \owner_\sigma(\rho) \land \sigma' = \sigma \cup\{e \rightarrow \rho\} \land \acyclic(\sigma')} \quad%\\[\infspacing]
	\infer[\rulename{release}]
		{\langle \pi : \rho.\syntax{release}(e) ; s,\sigma \rangle 
			\rightarrow 
		\langle \pi : s,\sigma \setminus \{e \rightarrow \rho\}  \rangle} 
		{ \{\pi\} = \owner_\sigma(e) \land \{e \rightarrow \rho\} \subseteq \sigma \land \{e' \rightarrow \rho\} \subseteq \sigma \land e' \neq e} \\[\infspacing]
	\infer[\rulename{spawn}]
		{\langle \pi : \pi' := \syntax{spawn}(s') ; s,\sigma \rangle 
			\rightarrow 
		\langle \pi : s | \pi' : s',\sigma \rangle} 
		{ \fresh_\sigma(\pi')} \quad%\\[\infspacing]
	\infer[\rulename{allocate}]
		{\langle \pi : \rho := e.\syntax{allocate} ; s,\sigma \rangle 
			\rightarrow 
		\langle \pi : s,\sigma \cup \{e \rightarrow \rho\} \rangle} 
		{ \fresh_\sigma(\rho)} \\[-6ex]
\end{gather*}
}
\vspace{-2ex}
\caption{Inference rules of SOM}
\label{fig:inference}
\end{figure*}

\begin{example}
  \figref{fig:pass-example} shows an ownership graph before and after the execution of process $\pi_1 : \rho_2.\syntax{pass}(\pi_1, \pi_2)$. Note that $\pi_1$ is the only root of $\rho_2$ and is therefore permitted to pass ownership. After executing the statement, the ownership edge $\pi_1 \rightarrow \rho_2$ has been removed, and the new edge $\pi_2 \rightarrow \rho_2$ has been inserted.
\end{example}

\begin{figure}[htb]
\vspace{-2ex}
\centering
\begin{subfigure}[b]{0.2\textwidth}
\begin{tikzpicture}[scale=0.75,transform shape]
  \Vertex[x=0,y=4,L=$\pi_2$]{p2}
  \Vertex[x=0,y=6,L=$\pi_1$]{p3}
  \Vertex[x=2,y=4,L=$\rho_3$]{r4}
  \Vertex[x=2,y=5,L=$\rho_2$]{r5}
  \Vertex[x=2,y=6,L=$\rho_1$]{r6}
  \Vertex[x=4,y=4,L=$\rho_{6}$]{r10}
  \Vertex[x=4,y=5,L=$\rho_{5}$]{r11}
  \Vertex[x=4,y=6,L=$\rho_{4}$]{r12}

  \tikzstyle{LabelStyle}=[fill=white,sloped]
//  \tikzstyle{EdgeStyle}=[bend left]
  \tikzstyle{EdgeStyle}=[post]

  \Edge[](p2)(r4)

  \Edge[color=blue](p3)(r5)
  \Edge[](p3)(r6)

  \Edge[](r4)(r10)

  \Edge[](r5)(r11)
  \Edge[](r5)(r12)
\end{tikzpicture}
\caption{}
\end{subfigure}
\qquad
\begin{subfigure}[b]{0.2\textwidth}
\begin{tikzpicture}[scale=0.75,transform shape]
  \Vertex[x=0,y=4,L=$\pi_2$]{p2}
  \Vertex[x=0,y=6,L=$\pi_1$]{p3}
  \Vertex[x=2,y=4,L=$\rho_3$]{r4}
  \Vertex[x=2,y=5,L=$\rho_2$]{r5}
  \Vertex[x=2,y=6,L=$\rho_1$]{r6}
  \Vertex[x=4,y=4,L=$\rho_{6}$]{r10}
  \Vertex[x=4,y=5,L=$\rho_{5}$]{r11}
  \Vertex[x=4,y=6,L=$\rho_{4}$]{r12}

  \tikzstyle{LabelStyle}=[fill=white,sloped]
//  \tikzstyle{EdgeStyle}=[bend left]
  \tikzstyle{EdgeStyle}=[post]
starts
  \Edge[](p2)(r4)

  \Edge[](p3)(r6)
  
  \Edge[color=red](p2)(r5)

  \Edge[](r4)(r10)

  \Edge[](r5)(r11)
  \Edge[](r5)(r12)
\end{tikzpicture}
\caption{}
\end{subfigure}
\vspace{-1ex}
\caption{Ownership graph before and after executing \mbox{$\pi_1 : \rho_2.\syntax{pass}(\pi_1, \pi_2)$}}
\label{fig:pass-example}
\end{figure}

Rule $\rulename{share}$ allows $\pi$ to share resource $\rho$ with entity $e$, provided $\pi$ is a root of $\rho$ and the sharing does not introduce cycles into the ownership graph. 

\begin{example}
  \figref{fig:share-example} shows an ownership graph before and after executing $\pi_1 : \rho_1.\syntax{share}(\rho_2)$. Note that $\pi_1$ is a root of $\rho_1$ and the insertion of a new edge from $\rho_2 \rightarrow \rho_1$ does not introduce cycles. After completion of the statement, this new edge is introduced into the graph.
\end{example}

\begin{figure}[htb]
\centering
\begin{subfigure}[b]{0.2\textwidth}
\begin{tikzpicture}[scale=0.75,transform shape]
  \Vertex[x=0,y=2,L=$\pi_2$]{p1}
  \Vertex[x=0,y=4,L=$\pi_1$]{p2}
  \Vertex[x=2,y=2,L=$\rho_2$]{r3}
  \Vertex[x=2,y=4,L=$\rho_1$]{r4}
  \Vertex[x=4,y=2,L=$\rho_5$]{r8}
  \Vertex[x=4,y=3,L=$\rho_4$]{r9}
  \Vertex[x=4,y=4,L=$\rho_{3}$]{r10}

  \tikzstyle{LabelStyle}=[fill=white,sloped]
//  \tikzstyle{EdgeStyle}=[bend left]
  \tikzstyle{EdgeStyle}=[post]
  \Edge[](p1)(r3)

  \Edge[](p2)(r4)

  \Edge[](r3)(r8)
  \Edge[](r3)(r9)

  \Edge[](r4)(r10)
\end{tikzpicture}
\caption{}
\end{subfigure}
\qquad
\begin{subfigure}[b]{0.2\textwidth}
\begin{tikzpicture}[scale=0.75,transform shape]
  \Vertex[x=0,y=2,L=$\pi_2$]{p1}
  \Vertex[x=0,y=4,L=$\pi_1$]{p2}
  \Vertex[x=2,y=2,L=$\rho_2$]{r3}
  \Vertex[x=2,y=4,L=$\rho_1$]{r4}
  \Vertex[x=4,y=2,L=$\rho_5$]{r8}
  \Vertex[x=4,y=3,L=$\rho_4$]{r9}
  \Vertex[x=4,y=4,L=$\rho_{3}$]{r10}

  \tikzstyle{LabelStyle}=[fill=white,sloped]
//  \tikzstyle{EdgeStyle}=[bend left]
  \tikzstyle{EdgeStyle}=[post]
  \Edge[](p1)(r3)

  \Edge[](p2)(r4)

  \Edge[](r3)(r8)
  \Edge[](r3)(r9)
  \Edge[color=red](r3)(r4)

  \Edge[](r4)(r10)
\end{tikzpicture}
\caption{}
\end{subfigure}
\vspace{-1ex}
\caption{Ownership graph before and after executing \mbox{$\pi_1 : \rho_1.\syntax{share}(\rho_2)$}}
\label{fig:share-example}
\end{figure}

Rule $\rulename{release}$ allows $\pi$ to release $e$'s ownership over $\rho$, provided $\pi$ is a root of $e$ and $e$ owns $\rho$.
Rule $\rulename{spawn}$ allows $\pi$ to spawn a fresh process $\pi'$ to execute statements $s$. Finally, rule $\rulename{allocate}$ allows $\pi$ to allocate a fresh resource $\rho$ owned by $e$.

\subsection{Rules preserve the graph properties}

We first establish that the inference system of \figref{fig:inference} is sound in the sense that the inference rules transform ownership graphs in a valid way.

\begin{lemma}
  Let $P$ be a process and $\sigma$ an ownership graph, and suppose that $\langle P,\sigma \rangle \rightarrow \langle P',\sigma' \rangle$. Then $\sigma'$ is an ownership graph.
\end{lemma}

\begin{proof}
We proceed by induction over the inference rules of SOM.

%\begin{itemize}
\noindent {$\rulename{par}_1$, $\rulename{par}_2$, $\rulename{read}$, $\rulename{write}$, $\rulename{spawn}$:}
These rules do not modify the graph and therefore preserve the graph properties.
\vspace{1ex}

\noindent {$\rulename{pass}$:}
Hence $\sigma' = (\sigma \setminus \{e \rightarrow \rho\}) \cup\{e' \rightarrow \rho\}$. The target $\rho$ of the new edge is not a process, therefore the property (P) of Definition~\ref{def:rcg} is preserved. The target $\rho$ of the old edge is the same as of the new edge, preserving (R). Furthermore, from the premise of the rule we have $\acyclic(\sigma')$, ensuring (A). 
\vspace{1ex}

\noindent {$\rulename{share}$:}
Hence $\sigma' = \sigma \cup\{e \rightarrow \rho\}$. The target $\rho$ of the new edge is not a process, therefore (P) is preserved. Adding an edge to a resource always preserves (R). Furthermore, from the premise of the rule we have $\acyclic(\sigma')$, ensuring (A).
\vspace{1ex}

\noindent {$\rulename{release}$:}
Hence $\sigma' = \sigma \setminus \{e \rightarrow \rho\}$. The removal of an edge trivially preserves (P) and (A). Property (R) is preserved because the premise $\{e' \rightarrow \rho\} \subseteq \sigma$, with $e' \neq e$, states that another edge with the same target $\rho$ has to exist in $\sigma$, and hence also in $\sigma'$.
\vspace{1ex}

\noindent {$\rulename{allocate}$:}
Hence $\sigma' = \sigma \cup \{e \rightarrow \rho\}$ where $\rho$ is fresh. This trivially preserves (P). Property (R) is preserved because the new node $\rho$ has an incoming edge. Property (A) is preserved because the new node $\rho$ has no outgoing edges. 

%\end{itemize}
\end{proof}

\subsection{Freedom from data races}
\label{sec:race-freedom}

The main property of SOM is freedom from data races. As demonstrated in \secref{sec:integration}, this property can be leveraged in a dynamic checking technique for safe resource access.

The occurrence of a data race is typically defined, e.g. in~\cite{savage1997eraser}, as follows: 
a \emph{data race} occurs when two processes access a shared resource and when
\begin{itemize}
\item at least one access is a write, and
\item the processes use no explicit mechanism to prevent the accesses from being simultaneous.
\end{itemize}

\noindent Within our formal framework, this definition can be expressed as follows:
\begin{definition}[Data race]
\label{def:race}
A \emph{data race} on a resource $\rho$ occurs when there exist processes $\pi_1$ and $\pi_2$, $\pi_1 \neq \pi_2$, and a state $\sigma$ such that $\pi_1: \rho.\syntax{write}$ is reducible in $\sigma$ and, in addition, either $\pi_2: \rho.\syntax{write}$ or $\pi_2: \rho.\syntax{read}$ is reducible \mbox{in $\sigma$.}
\end{definition}

\begin{theorem}
SOM is free of data races.
\end{theorem}

\begin{proof}
In order for $\pi_1: \rho.\syntax{write}$ to be reducible in a state $\sigma$, it has to fulfill the premise of the $\rulename{write}$ rule, which equals to $\{\pi_1\} = \owner_\sigma(\rho)$. 
Furthermore, in order for $\pi_2: \rho.\syntax{write}$ to be reducible in a state $\sigma$, it has to fulfill  $\{\pi_2\} = \owner_\sigma(\rho)$; in order for $\pi_2: \rho.\syntax{read}$ to be reducible in a state $\sigma$, it has to fulfill the premise of the $\rulename{read}$ rule, which equals to $\pi_2 \in \owner_\sigma(\rho)$. 
We conclude in both cases that no data race can occur as $\pi_2 \not \in \{\pi_1\} = \owner_\sigma(\rho)$.
\end{proof}

\section{Dynamic checking of ownership violations}
\label{sec:integration}

SOM can be combined with a regular programming language to provide a dynamic technique to find ownership violations. Since a checking technique may not change the semantics of the program, a thread does not wait until the premise of a SOM statement is met. Instead, it treats the premise as an assertion, triggering an error if it is not met.

The approach comprises the following steps:
\begin{enumerate}
\item Add SOM statements to the programming language.
\item Augment existing language instructions with the new SOM statements.
\item Adapt synchronization mechanisms to issue SOM statements.
\item \begin{enumerate}
\item In development mode, run the SOM model together with the program, raising errors if the model cannot take a step.
\item In production mode, ignore all SOM statements when running.
\end{enumerate}
\end{enumerate}
This section describes each step. We use an object-oriented programming language (Java) for integration of SOM statements as it enables us to map objects to resources in the ownership graph. %The last part of this section mentions some considerations for the implementation.

\subsection{Adding SOM statements}
\label{sec:instructions}

%TODO: Update this to reflect current implementation
All SOM statements, except for $\syntax{spawn}$, are implemented as methods of all \emph{regular classes}, i.e.\ in a common superclass. \emph{Non-regular classes} are synchronization mechanisms (locks, channels, etc.), value types, immutable types and legacy code that is not ready for SOM; non-regular classes are not checked using SOM and therefore do not need to be modified. Objects of non-regular types cannot serve as owners, with the exception of adapted synchronization mechanisms (see \secref{sec:synchronization}).
The $\syntax{spawn}$ statement, in programming languages similar to Java with a thread class and a separate method to start the thread, can be introduced as a hidden static method of the thread class. 

All these methods serve as hooks for the simulation environment and do not change the state of the program. The $\syntax{pass}$ statement is the only statement directly available to the programmer, all others are hidden and only triggered by regular instructions (section \ref{sec:mapping}) or synchronization mechanisms (section \ref{sec:synchronization}).

\begin{example}
Listing \ref{lst:llsplit} shows when to use the \java{pass} statement in real code: a linked list is split, which requires that the nodes which belong to the new linked list are also owned by it. This is achieved while iterating through the list up to the split point. In this case, the \java{pass} statement does not pass ownership to another process, but it reorganizes ownership within the part of the sub-graph owned by the current process. The code does not feature any synchronization mechanisms, however, if multiple threads attempt to split the same list, the threads without write access trigger an error.
\end{example}

\begin{Listing}[t]
\vspace{-3ex}
\begin{lstlisting}
class LinkedList {
	protected LLNode first;
	public LinkedList splitBefore(int index) {
		LinkedList result = new LinkedList();
		result.first = first;
		//Transfer ownership of first node to the new list
		result.first.pass(this, result);
		//Find split point
		LLNode cur = first;
		for (int i = 1; i < index; i++) {
			cur = cur.next;
			//Transfer ownership of current node to the result
			cur.pass(this, result);
		}
		first = cur.next;
		//Make the cut, all further nodes are not affected
		cur.next = null;
		return result;
	}
}
\end{lstlisting}
\vspace{-2ex}
\caption{Splitting a linked list}
\label{lst:llsplit}
\end{Listing}

\subsection{Augmenting language instructions}
\label{sec:mapping}

This section explains when a regular language instruction is augmented by a SOM statement. \tabref{tab:mapping} gives a summary, showing the Java instructions or expressions on the left and the additionally executed SOM statements on the right. The SOM statements are invoked immediately before the regular instruction.

\begin{savenotes}
\begin{table}[htb]
{\small
\begin{center}
\begin{tabular}{l|l}
\textbf{Program code} & \textbf{SOM statements}\\
\hline
\sjava{... object.fieldname ...} & $\mathrm{object}.\syntax{read}$ \\
\hline
\sjava{object.fieldname = ...;} & $\mathrm{object}.\syntax{write}$ \\
\hline
\sjava{o = new RegularClass(...);} & $o := $ \sjava{this}$.\syntax{allocate}$ \\
\hline
\sjava{((Thread)t).start();} & $\pi:=\syntax{spawn(s)};$\footnote{Here, and in following uses of \syntax{spawn}, the statement $s$ contains all SOM statements that the thread \syntax{t} is issuing.}\\
                            & $t.\syntax{pass}($\sjava{this}$,\pi)$ \\
\end{tabular}
\end{center}
}
\vspace{-2ex}
\caption{Augmenting language instructions}
\label{tab:mapping}
\end{table}
\end{savenotes}

\paragraph{Field access.}
Retrieving the value of a field of an object is accompanied with a SOM \texttt{read} statement targeted to the object. Setting a field is similar, but issues a \texttt{write}. We therefore insert an \java{object.read} before every reading field access and an \java{object.write} before every writing access.

\paragraph{Creating a new object.}
\label{sec:firstrecowns}
Whenever the program creates a new object from a regular class, an \texttt{allocate} statement is issued. The target of this statement is a staging object which owns all new objects created by the current process. When an object owned by the staging object is assigned to a field, the field's owner becomes the owner of the object. We call this the \emph{first receiver owns} principle. This enables factories and other common patterns without the need of explicit \texttt{pass} statements. It is important that the \texttt{allocate} statement is issued before the constructor of the new object is run, otherwise the constructor would raise false ownership violation alarms.

\paragraph{Spawning threads and allocating objects.}
In Java, the constructor of a thread is executed by the original thread, which has to be respected: before a thread is running, it is a regular object and therefore has an owner. Only executing the \java{start()}-Method actually creates a new thread, at which point the resource that formerly denoted the thread object becomes a process in the SOM model.

We model this behavior by declaring that the \java{Thread} object is just a regular object with the process it represents later being transparent. When the \java{start()} method is executed, \syntax{spawn} is executed and the resulting process is given the thread object. This enables initialization code to be run by the thread that creates the \java{Thread} object. It also means that every thread will only directly own one resource. All other resources owned by the thread are owned indirectly through this one resource.

\subsection{Adapting synchronization mechanisms}
\label{sec:synchronization}

This section explains basic synchronization mechanisms and how they interact with the SOM model. Synchronization mechanisms are not considered regular objects, so they are not represented as resources in the ownership graph. Instead, they are represented as processes because synchronization mechanisms must be able to acquire ownership over resources and are not owned themselves. 

Therefore the model considers all code that is executed within the context of the synchronization mechanism as being executed by its representing process, even if the actual code is executed by the calling thread. 

In this section, we sometimes refer to the special variable \java{caller} which denotes the object owning the context in which the call statements originated, i.e.\ the \java{this} variable of the previous stack frame.

\paragraph{Channels.}
\label{sec:channels}

\tabref{tab:channels} shows how the operations on channels are expressed in SOM. 

\begin{savenotes}
\begin{table}[htb]
\begin{center}
{\small
\begin{tabular}{l|l}
\textbf{Program code} & \textbf{SOM statements} \\ 
\hline 
\sjava{c = new Channel();} & $c := \syntax{spawn}(s)$ \\
\hline 
\sjava{c.send(o);} & $o.\syntax{pass}($\sjava{this}$, c)$ \\
\hline
\sjava{o = c.receive();} & $o.\syntax{pass}(c, \mathrm{caller})$\footnote{$\syntax{pass}$ is issued by $c$}\\
\end{tabular}
}
\end{center}
\vspace{-2ex}
\caption{Channels}
\label{tab:channels}
\end{table}
\end{savenotes}

Passing a message (essentially an object) through a channel transfers ownership from the sender to the receiver. Therefore, a thread can only put objects it owns into channels, thereby releasing ownership over them. A process calling the \java{send(message)} method automatically passes ownership over the message from the current object to the channel using a SOM $\syntax{pass}$ statement. A process receiving a message from a channel through the \java{receive()} method also receives ownership over it: the channel passes ownership over the object \java{o} using the SOM $\syntax{pass}$ statement to the object in whose context the \java{receive()} was executed.

\begin{Listing}[htb]
\vspace{-3ex}
\begin{lstlisting}
class Stage implements Runnable {
	private final Channel<WorkItem> input, output;
	public Stage(Channel<WorkItem> input, Channel<WorkItem> output) {
		this.input = input;
		this.output = output;
	}
	public void run() {
		while (true) {
			WorkItem p = input.receive();
			p.data = p.data + 1;
			output.send(p);
		}
	}
}
\end{lstlisting}
\caption{Pipeline using channels}
\label{lst:pipeline}
\vspace{-2ex}
\end{Listing}

\begin{example}
\label{sec:pipelines}
Listing \ref{lst:pipeline} shows the code required for a pipeline stage. The pipeline stages are connected using channels. 

We  assume a pipeline with $n$ stages handled by threads denoted as $\pi_1$ to $\pi_n$ and look at two threads, $\pi_k$ and $\pi_{k+1}$, with $1<k<n-2$. The channels connecting the stages are denoted as $c_{l}$, where $c_{l}$ connects $\pi_l$ and $\pi_{l+1}$. Work items $\rho$ are numbered from $1$ to $m$. We look at two work items $\rho_i$ and $\rho_{i+1}$, with $1 \le i < m$. 

We prefix the thread identifier $\pi$ before every instruction both in the Java code and the SOM. We use $\tau$ to denote a (thread local) stack value. The trace of a possible execution is shown in \tabref{tab:channelexec}.
\end{example}

\begin{table}[htb]
%\vspace{-2ex}
\centering
{\small
\begin{tabular}{@{}l@{}|@{\hspace{2pt}}l@{}}
\textbf{Program code} & \textbf{SOM statements} \\
\hline
$\pi_k$: $\rho_i$ = input.receive() & $c_{k-1}: \rho_{i}.\syntax{pass}(c_{k-1},\pi_k)$\\
$\pi_k$: $\tau$ = $\rho_i$.data & $\pi_k: \rho_i.\syntax{read}$ \\
$\pi_k$: $\rho_i$.data = $\tau$ + 1& $\pi_k: \rho_i.\syntax{write}$ \\
$\pi_k$: output.send($\rho_i$) & $\pi_k: \rho_i.\syntax{pass}(\pi_{k},c_{k})$ \\
$\pi_{k+1}$: $\rho_i$ = input.receive() & $c_{k}: \rho_{i}.\syntax{pass}(c_{k},\pi_{k+1})$\\
$\pi_k$: $\rho_{i+1}$ = input.receive() & $c_{k-1}: \rho_{i+1}.\syntax{pass}(c_{k-1},\pi_k)$\\
$\pi_k$: $\tau$ = $\rho_{i+1}$.data & $\pi_k: \rho_{i+1}.\syntax{read}$ \\
$\pi_{k+1}$: $\tau$ = $\rho_i$.data & $\pi_{k+1}: \rho_i.\syntax{read}$ \\
$\pi_k$: $\rho_{i+1}$.data = $\tau$ + 1& $\pi_k: \rho_{i+1}.\syntax{write}$ \\
$\pi_{k+1}$: $\rho_i$.data = $\tau$ + 1& $\pi_{k+1}: \rho_i.\syntax{write}$ \\
$\pi_{k+1}$: output.send($\rho_i$) & $\pi_{k+1}: \rho_i. \syntax{pass}(\pi_{k+1}, c_{k+1})$ \\
$\pi_k$: output.send($\rho_{i+1}$) & $\pi_k: \rho_i.\syntax{pass}(\pi_k, c_k)$ \\
$\pi_{k+1}$: $\rho_{i+1}$ = input.receive() & $c_{k}: \rho_{i+1}.\syntax{pass}(c_{k},\pi_{k+1})$\\
$\pi_{k+1}$: $\tau$ = $\rho_{i+1}$.data & $\pi_{k+1}: \rho_{i+1}.\syntax{read}$ \\
$\pi_{k+1}$: $\rho_{i+1}$.data = $\tau$ + 1& $\pi_{k+1}: \rho_{i+1}.\syntax{write}$ \\
$\pi_{k+1}$: output.send($\rho_{i+1}$) & $\pi_{k+1}: \rho_{i+1}.\syntax{pass}(\pi_{k+1}, c_{k+1})$
\end{tabular}
}
\caption{Pipeline example: execution trace}
\label{tab:channelexec}
\end{table}

%\vspace{-2ex}
\paragraph{Queues.}
\label{sec:queues}
A queue can also be used to synchronize threads. In this case, messages residing inside the queue are not accessible to any thread. A thread peeking at the head of the queue gets a reference to the object but may not access its state. Except for the asynchronous behavior, a queue can be expressed in SOM similarly to a channel, see \tabref{tab:queues}.

\begin{savenotes}
\begin{table}[htb]
\begin{center}
{\small
\begin{tabular}{l|l}
\textbf{Program code} & \textbf{SOM statements} \\ 
\hline 
\sjava{q = new Queue();} & $q :=$ \sjava{this}$.\syntax{spawn}(s)$ \\
\hline 
\sjava{q.add(o);} & $o.\syntax{pass}($\sjava{this}$, q)$ \\
\hline
\sjava{o = q.receive();} & $o.\syntax{pass}(q, \mathrm{caller})$\footnote{$\syntax{pass}$ is issued by $q$}\\
\end{tabular}
}
\end{center}
\vspace{-2ex}
\caption{Queues}
\label{tab:queues}
\vspace{-2ex}
\end{table}
\end{savenotes}

\paragraph{Locks.}
\label{sec:locks}
Monitors, mutexes and semaphores are called locks within SOM. A lock is permanently associated with a resource upon creation of the lock, similar to monitors \cite{Hoare:1974:MOS:355620.361161}. It differs from a common monitor as it is perfectly possible to have multiple locks protecting the same object: it suffices if at most one of these locks is not locked.

We differentiate between two kinds of locks: regular locks based on monitors and binary semaphores. The latter allow the thread releasing the lock to be different from the thread acquiring it. On the other hand, monitors allow the usage of condition variables. \tabref{tab:locks} shows how lock operations are expressed in SOM.

The \java{lock()} method acquires the lock and transfers ownership of the object from the lock to the caller using $\syntax{pass}$. 
The \java{unlock()} method releases the lock and transfers ownership over the content from the caller to the lock using $\syntax{pass}$, therefore the caller needs to have ownership of the resource prior to the unlocking operation. 
Releasing and acquiring the lock because of a condition variable works exactly the same as releasing and acquiring the lock directly, including all the resource passing. 

\begin{savenotes}
\begin{table}[htb]
\begin{center}
{\small
\begin{tabular}{l|l}
\textbf{Program code} & \textbf{SOM statements} \\ 
\hline
\sjava{l = new Lock(o);} & $l := \syntax{spawn}(s);$ \\
& $o.\syntax{pass}($\sjava{this}$,l)$ \\
\sjava{c = l.newCondition();} & \\
\hline 
\sjava{l.lock();} & $o.\syntax{pass}(l,\mathrm{caller})$\footnote{$\syntax{pass}$ is issued by $l$}\\
\hline
\sjava{l.unlock();} & $o.\syntax{pass}($\sjava{this}$,l)$ \\
\hline
\sjava{c.wait();} & $o.\syntax{pass}($\sjava{this}$,l)$ \\
``received lock again" & $o.\syntax{pass}(l,\mathrm{caller})$\footnote{$\syntax{pass}$ is issued by $l$}\\
\sjava{c.signal();}\footnote{We assume a signal-and-continue signaling discipline} &\\
\end{tabular}
}
\end{center}
\vspace{-2ex}
\caption{Locks}
\label{tab:locks}
\vspace{-2ex}
\end{table}
\end{savenotes}

\begin{example}
\label{sec:pingpong}
The code in listing \ref{lst:pingpong} shows two binary semaphores both holding the same object. According to the semantics of locks, it follows that only one lock may legally be unlocked at any point in time. In the initial state, \java{ping} is locked and \java{pong} is unlocked. The two threads repeatedly unlock one of the locks and then lock the other one. It suffices that the unlocking thread currently owns the object.
\end{example}

\begin{Listing}[t]
\vspace{-3ex}
\begin{lstlisting}
class PingPong {
	public PingPong() {
		Ball b = new Ball();
		BinSemaphore<Ball> ping = new BinSemaphore<>(b);
		//ping owns b, we want the ownership back, so we need to lock ping
		ping.lock();
		BinSemaphore<Ball> pong = new BinSemaphore<>(b);
		//pong owns b
		(new Player(ping, pong)).start();
		//This player will lock pong, starting the game
		(new Player(pong, ping)).start();
	}
}
class Player extends SafeThread {
	public final BinSemaphore<Ball> ping, pong;	
	public Player(BinSemaphore<Ball> ping, pong) {
		this.ping = ping;
		this.pong = pong;
	}
	void run() {
		while (true) {
			//Take the ball from the lock
			ping.lock();
			//Let it bounce once
			ping.getResource().bounce();
			//Pass the ball to the other player
			pong.unlock();
		}
	}
}
class Ball {
	int bounces = 0;
	public void bounce() { bounces++; }
}
\end{lstlisting}
\vspace{-2ex}
\caption{Locks: Ping Pong example}
\label{lst:pingpong}
\end{Listing}

%\begin{example}
%The code in listing \ref{lst:philosophers} shows the dining philosophers and uses locks holding fork objects. This implementation is naive and can produce the classic Coffman deadlock if none or all of the Philosophers receive swapped locks.
%\end{example}
%
%\begin{Listing}
%\begin{lstlisting}
%class Philosopher implements Runnable {
%	private final Lock<Fork> left, right;
%	private final int iterations = 1000;
%	private final Random random;
%	
%	public Philosopher (Lock<Fork> left, Lock<Fork> right, long number) {
%		this.left = left; this.right = right;
%		random = new Random(number);
%	}
%	
%	public void run () {
%		for (int i = 0; i < iterations; i++) {
%			think(); eat();
%		}
%	}
%	
%	private void eat () {
%		left.open (); right.open ();
%		left.getResource().use(); right.getResource().use();
%		left.close (); right.close ();
%	}
%	
%	private void think () {
%	}
%}
%
%class Fork extends SafeObject {
%	private long wear;
%	
%	public void use() {
%		wear++;
%	}
%}
%\end{lstlisting}
%
%\caption{Locks: dining philosophers example}
%\label{lst:philosophers}
%\end{Listing}

%\vspace{-3ex}
\paragraph{Readers-writer locks.}
\label{sec:rwlocks}
We introduce a readers-writer lock similar to the regular locks. In addition to regularly acquiring and releasing the lock, it also supports acquiring for reading. We use a proxy object \java{p} to manage multiple owners. 

\tabref{tab:rwlocks} shows how the operations on readers-writer locks are expressed in SOM. Calling \java{lockWrite()} acquires the write lock and makes the lock pass the protected object from the proxy to the caller. Calling \java{unlockWrite()} releases the write lock and uses $\syntax{pass}$ to transfer ownership from the caller to the proxy. Calling \java{lockRead()} acquires the read lock and makes the lock add the caller to the owners of the proxy using $\syntax{share}$. Calling \java{unlockRead()} releases the read lock and removes the caller from the list of owners of the proxy using $\syntax{release}$. The usage of the proxy object enables us to only have a single class that needs to handle multiple owners (see section \ref{sec:multicontroller}).

\begin{savenotes}
\begin{table}[htb]
\begin{center}
{\small
\begin{tabular}{l|l}
\textbf{Program code} & \textbf{SOM statements} \\ 
\hline
\sjava{l = new ReadWriteLock(o);} & $l :=$ $\syntax{spawn}(s);$  \\
& $p := l.\syntax{allocate};$ \\
& $o.\syntax{pass}($\sjava{this}$,p)$ \\
\hline 
\sjava{l.lockWrite();} & $o.\syntax{pass}(p, \mathrm{caller})$\footnote{$\syntax{pass}$ is issued by $l$}\\
\hline
\sjava{l.unlockWrite();} & $o.\syntax{pass}($\sjava{this}$, p)$ \\
\hline 
\sjava{l.lockRead();} & $p.\syntax{share}(\mathrm{caller})$\footnote{$\syntax{share}$ is issued by $l$}\\
\hline 
\sjava{l.unlockRead();} & $p.\syntax{release}($\sjava{this}$)$ \\
\end{tabular} 
}
\end{center}
\vspace{-2ex}
\caption{Readers-writer locks}
\label{tab:rwlocks}
\end{table}
\end{savenotes}

\subsection{Running SOM-augmented programs}
\label{sec:simulation}

\paragraph{Development.}
At development time, it is important to know if the program improperly accesses shared objects. When the program executes a SOM statement whose premise is not fulfilled, an \emph{ownership violation} occurs. An ownership violation can have two reasons:
\begin{enumerate}
\item A resource is improperly accessed.
\item A statement modifying the ownership graph violates the graph properties.
\end{enumerate}

Detecting ownership violations is achieved by simulating the SOM model together with a run of the program. The simulation is synchronized on the calls to the SOM statements introduced into the programming language in \secref{sec:instructions}. 
If the SOM model is unable to take a step together with the program, the program contains an ownership violation. 
The simulation is \emph{complete} in the sense that all errors reported are actual ownership violations. %A specific ownership violation, however, does not necessarily produce a data race: if, for example, semaphore synchronization is used correctly to prevent a data race, but the ownership has not been properly passed, an ownership violation is reported. 
As with most dynamic techniques, the simulation is \emph{not sound}. Code that is not executed is not checked, and it is possible that certain interleavings can lead to ownership violations being missed. 

\paragraph{Deployment.}
At deployment time, the program can be run without simulating the SOM model by ignoring the SOM statements. This removes the overhead of the simulation and thereby improves performance.

\section{Implementation and performance analysis}
\label{sec:performanceanalysis}

This section describes implementation details of the integration of SOM into Java, and benchmark results.% on this prototype, investigating the overhead of synchronization and field access.

\subsection{Java SOM implementation}
\label{sec:implementation}

The SOM simulation can be integrated tightly into the Java runtime, i.e.\ the simulation code is implemented directly in the methods that have been introduced for SOM. This enhances performance and can make it easier to locate errors, as information such as stack traces exist. By using the assertion facility offered by Java, we can even focus on parts of a program (assertions may be enabled for specific classes).

\paragraph{Representation of the ownership graph.}
The ownership graph does not require its own data structure, but can be represented using regular objects by appending a field pointing to the owner(s). This increases the performance of the simulation considerably. This approach also keeps the representation of the graph small, as the memory management takes care of removing vertices no longer used.

\paragraph{Restriction to single owners.}
\label{sec:multicontroller}
By providing a separate class that can have multiple owners, regular classes only need a single field for the owner. A readers-writer lock (see section \ref{sec:rwlocks}) can then use this special class to manage the sharing. %This way, a programmer is just presented with the lock which she or he can use to share the object.

This has also a positive effect on usability: to identify the exact edge to redirect, the $\syntax{pass}$ statement of SOM takes the old and the new owners as arguments. When only one special class has multiple owners, we can drop the first argument of $\syntax{pass}$ for all regular classes, and are still able to identify the exact edge to modify. This simplifies the usage of synchronization mechanisms: it is no longer necessary for the caller to own the object, it suffices if the current thread has the necessary ownership.

\paragraph{Late checking of graph properties.}
The actual check for the acyclicity of the ownership graph can be dropped since it will show up during the depth-first search for the root owner. Cycles in the graph result in infinite searches, and an empty owner set results in an immediate failure. This optimization trades the information of the exact position where a graph became degenerated for vastly increased performance. 
% Since every SOM statement needs to initiate the (depth-first) search for a root owner anyway, degenerated graphs can be detected easily as part of this step: cycles result in an infinite search for the owner, whereas a resource with no incoming edges results in an empty set of owners. While using this technique makes it more difficult to locate the exact position where the graph became degenerated, it results in a significant performance gain.

\paragraph{Prototype.}
The current prototype for our implementation is freely available for download\footnote{SOM homepage: \url{http://se.inf.ethz.ch/people/schill/som}}.
The prototype implements SOM checking for Java using bytecode transformation. The support comes in form of a library, which provides the synchronization mechanisms explained in \secref{sec:synchronization} as well as wrappers for some important classes of the Java base library.

%The library transforms all classes supporting SOM when loaded into the virtual machine. This overhead can be mitigated by modifying the class files directly, but we do not support this at the moment.

The transformer is informed either through command line arguments or through annotations which classes to transform. Threads are automatically considered processes. We have a class called \java{SharedReader} which can handle multiple owners, all other classes can only have one owner for efficiency and convenience. At the moment, \java{SharedReader} is only used for the readers-writer lock.

The transformer injects the SOM methods when the class is loaded, implementing the \java{Resource} and \java{Process} interfaces. Since the methods are not available at compile time and in production environments, the \texttt{pass} operation is available as a static method of the \java{SOM} class.
\subsection{Benchmark results}
\label{sec:performance}

Simulating SOM incurs an overhead as it requires checks whenever a field of an object is accessed. Our primary goal was to make sure that a program which uses the production version is not significantly slower than a regular program. Our secondary goal was to have a reasonable overhead when SOM is activated, such that it is practical to use during development. For the second goal we implemented the checks so that they can be disabled on a class-by-class basis, using the regular mechanisms for Java assertions.

We created three benchmarks to assess the performance of our implementation. The benchmarks have been executed with full, partial and no simulation of the SOM model, as well as a version that uses plain Java. In the legends of Figures~\ref{fig:performancepp} and~\ref{fig:performancesort}, these benchmark configurations are denoted as follows:
\begin{itemize}
\item \emph{Full}: Full SOM simulation.
\item \emph{Partial}: Constructing the ownership graph, no access checks.
\item \emph{None}: Ignoring all SOM statements.
\item \emph{Base}: Pure Java with locks.
\end{itemize}
Partial simulation disables all $\syntax{read}$ and $\syntax{write}$ statements, whereas no support uses our production library that strips out all SOM statements. We used a quad core machine to run the benchmarks, taking the average time of 100 runs in all cases.

\paragraph{Synchronization overhead.}
To asses the impact of SOM on the synchronization mechanisms, we ran the Ping Pong example of Listing~\ref{lst:pingpong}: in quick succession, two threads lock and unlock two semaphores that protect a single object.%, the ball.

The results of this benchmark are shown in \figref{fig:performancepp} and \tabref{tab:performancepp}. Since the majority of the SOM code in this example concerns passing ownership, the difference between full and partial support is minimal. Furthermore, the difference between the production library and our base line is negligible and can be attributed to abstraction of locks, fulfilling our primary goal. Overall the impact of SOM on this example is small.

\begin{figure}[t]
  \centering
\begin{tikzpicture}
\begin{axis}[scaled ticks=false,xlabel=Locking/Unlocking pairs,ylabel=Time (milliseconds),height=4.8cm,width=\columnwidth,legend pos=north west]
\addplot plot[error bars/.cd, y dir=both,y explicit] table [x=parameter, y=time, col sep=comma] {pp.full.csv};
\addlegendentry{Full}
\addplot plot[error bars/.cd, y dir=both,y explicit] table [x=parameter, y=time, col sep=comma] {pp.light.csv};
\addlegendentry{Partial}
\addplot plot[error bars/.cd, y dir=both,y explicit] table [x=parameter, y=time, col sep=comma] {pp.none.csv};
\addlegendentry{None}
\addplot plot[error bars/.cd, y dir=both,y explicit] table [x=parameter, y=time, col sep=comma] {pp.base.csv};
\addlegendentry{Base}
\end{axis}
\end{tikzpicture}
  \vspace{-5ex}
  \caption{Ping Pong benchmark: performance of locks}
  \label{fig:performancepp}
\end{figure}
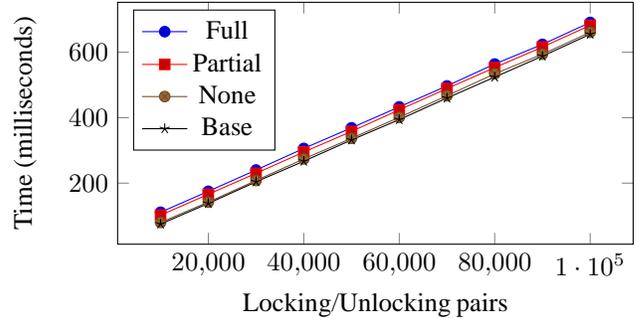

\paragraph{Field access overhead.}
Our second benchmark is concerned with the overhead the SOM checks have on field accesses.  This also gives an impression on how much speedup we can gain if checks in uninteresting parts of the program are disabled. The benchmark is running a single-threaded Quicksort algorithm; it is single-threaded so that we do not measure the time needed for synchronization, which is the main concern of our first benchmark.

The results of the sorting benchmark are shown in \figref{fig:performancesort} and  \tabref{tab:performancesort}. In this case, the difference between activated (Full) and deactivated (Partial) assertions is noticeable, while the difference between Partial and None remains small. %This shows that when selectively using assertions for checking some part of the code base can be very fast. 
But even with all the SOM checks switched on, the program takes less than 2.3 times longer than the base version for all checked array sizes, fulfilling our secondary performance goal. As with the last benchmark, the difference between our production library and the baseline is negligible, fulfilling our primary performance goal.

\begin{figure}[t]
  \centering
\begin{tikzpicture}
\begin{axis}[scaled ticks=false,xlabel=Size of the Array,ylabel=Time (milliseconds),height=4.8cm,width=\columnwidth,legend pos=north west]
\addplot plot[error bars/.cd, y dir=both,y explicit] table [x=parameter, y=time, col sep=comma, y error=error] {sort.full.csv};
\addlegendentry{Full}
\addplot plot[error bars/.cd, y dir=both,y explicit] table [x=parameter, y=time, col sep=comma, y error=error] {sort.light.csv};
\addlegendentry{Partial}
\addplot plot[error bars/.cd, y dir=both,y explicit] table [x=parameter, y=time, col sep=comma, y error=error] {sort.none.csv};
\addlegendentry{None}
\addplot plot[error bars/.cd, y dir=both,y explicit] table [x=parameter, y=time, col sep=comma, y error=error] {sort.base.csv};
\addlegendentry{Base}
\end{axis}
\end{tikzpicture}
  \vspace{-5ex}
  \caption{Quicksort benchmark: performance of checks}
  \label{fig:performancesort}
\end{figure}
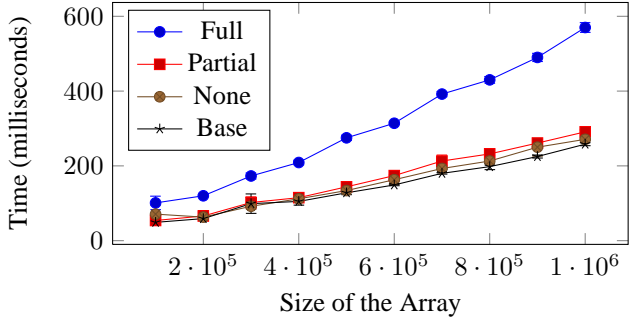

\paragraph{General performance and scaling.}
Our third benchmark investigates the performance in a complex algorithm. We developed a concurrent version of the Delaunay refinement benchmark of the Lonestar benchmark suite~\cite{kulkarni2009lonestar} and ran it using SOM. The result of this benchmark are shown in \figref{fig:performancedelref} and  \tabref{tab:performancedelref}. We used the massive (250k) dataset of the Lonestar benchmark suite and ran it with different numbers of workers. 

The results show that in real programs with nontrivial parallelism, the overhead of checking is reasonable. This example involves a great number of small objects, therefore the difference between no and partial SOM support is noticeable due to the higher demand of memory. We did not create a pure Java version, since our previous benchmarks show almost no difference between \emph{none} and \emph{base}.

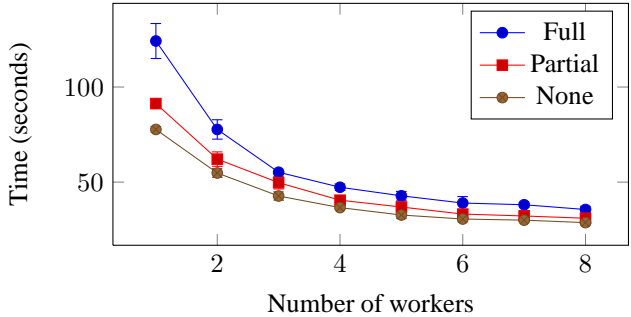
\begin{figure}[t]
  \vspace{-2ex}
  \centering
\begin{tikzpicture}
\begin{axis}[scaled ticks=false,xlabel=Number of workers,ylabel=Time (seconds),height=4.8cm,width=\columnwidth,legend pos=north east]
\addplot plot[error bars/.cd, y dir=both,y explicit] table [x=parameter, y=time, col sep=comma, y error=error] {delref.full.csv};
\addlegendentry{Full}
\addplot plot[error bars/.cd, y dir=both,y explicit] table [x=parameter, y=time, col sep=comma, y error=error] {delref.light.csv};
\addlegendentry{Partial}
\addplot plot[error bars/.cd, y dir=both,y explicit] table [x=parameter, y=time, col sep=comma, y error=error] {delref.none.csv};
\addlegendentry{None}
%\addplot plot[error bars/.cd, y dir=both,y explicit] table [x=parameter, y=time, col sep=comma, y error=error] {delref.base.csv};
%\addlegendentry{Base}
\end{axis}
\end{tikzpicture}
  \vspace{-5ex}
  \caption{Delaunay refinement benchmark: performance}
  \label{fig:performancedelref}
\end{figure}

\begin{table}[htb]
\noindent
\begin{subtable}{.5\textwidth}
  \centering
  {\small
  \include{pp.table} 
  }
  \vspace{-1ex}
  \caption{Ping Pong benchmark}\label{tab:performancepp}
  \vspace{1ex}
\end{subtable}
\qquad
\begin{subtable}{.5\textwidth}
  \centering
  {\small
  \include{sort.table}
  }
  \vspace{-1ex}
  \caption{Quicksort benchmark}\label{tab:performancesort}
  \vspace{1ex}
\end{subtable}
\qquad
\begin{subtable}{.5\textwidth}
  \centering
  {\small
  \include{delref.table}
  }
  \vspace{-1ex}
  \caption{Delaunay refinement benchmark}\label{tab:performancedelref}
  \vspace{1ex}
\end{subtable}
\vspace{-2ex}
\caption{Performance measurements}
\label{tab:measurements}
\end{table}

\subsection{Coding overhead of shared ownership}
To be practical, it is desirable that SOM incurs only a minor coding overhead. \tabref{tab:programs} gives an impression on the number of SOM statements required for checking a program. Whereas the simple examples need none, even the more complex Delaunay refinement, including the graph classes, just need a single SOM statement. The provided synchronization measures take care of most of the ownership transfers and our principle of \emph{first receiver owns} (see \secref{sec:firstrecowns}). In addition to our own examples, we also converted NanoHTTPD~\cite{nanohttpd}, which does not encapsulate the state in a way that makes SOM checking easy. However, we still only needed to modify 1.5\% of the code in order to pass the test cases, thereby also fixing some synchronization issues.

\begin{table}[htb]
\centering
\begin{tabular}{lcc}
  \hline
Name & lines of code & \#SOM statements\\
  \hline
monitor & 48 & 0 \\
quicksort & 69 & 0 \\
ping pong & 71 & 0 \\
car sharing & 78 & 0 \\
dining philosopher & 79 & 0 \\
delaunay refinement & 1187 & 1 \\
nanohttpd & 904 & 14 \\
  \hline
\end{tabular}
\caption{SOM programs}
\label{tab:programs}
\end{table}

\section{Related work}
\label{sec:related-work}
%TODO: Mention Data Parallel Java and its extensions

%%Ensuring freedom of race conditions can be achieved by associating resources with synchronization mechanisms (e.g.\ through monitors~\cite{Hoare:1974:MOS:355620.361161}). However, this does not work well with aggregate objects in object oriented programming languages. There is some literature on how to derive locks and critical sections automatically from annotations over atomic sets~\cite{}. These approaches are unsuitable when explicit synchronization and communication is warranted. 

The ownership methodology~\cite{Clarke:1998:OTF:286936.286947} has proved to be successful in reasoning about multithreaded programs before. 
Boyapati et al.~\cite{Boyapati:2002:OTS:582419.582440} define a static type system to ensure freedom from data races and deadlocks, and specify ownership relations between threads and objects similar to our notion of shared ownership. %Our approach allows more programming flexibility, however, since the ownership graph can dynamically change whereas the ownership relation is static. Furthermore, approaches based on type systems require some annotation overhead, whereas the SOM approach is very lightweight in this respect.
Jacobs et al.~\cite{jacobs2005safe} present a technique for static verification of concurrent programs, based on a programming methodology for handling aggregate object structures and their invariants using an ownership system~\cite{barnett:2004:verification}. In contrast, shared ownership allows an object/resource to be owned by multiple resources or processes, and ensures that it is always owned by at least one resource or process. These two requirements force the writer of a program to specify how the ownership of resources changes within the program; this specification can be completely dynamic. Our dynamic checking technique does not require fields being marked as referencing a part of the aggregate or a specific type system. 
%However, if and when a evaluation of an invariant, just as any expression, involves the state of a referenced object, this object has to be owned by the object defining the invariant, since it is considered a read access.

In our approach, every process that has access to a resource can share it with others for reading, this is similar to fractional permissions ~\cite{boyland2003checking}, where a process having (a fraction of) a permission can divide it further for giving access to more processes. Only if a process has full permission, it can access the resource for writing.

Martin et al.~\cite{Martin:2010} present an approach to checking ownership policies dynamically in concurrent C/C++ programs, similar to our dynamic technique. However, in order to read memory in parallel, it has to be marked read-only and this can not be reverted. Handling temporarily shared resources for reading is a core part of our approach by using a graph to represent the relationships between processes and resources. 

Apart from ownership-based approaches, other programming techniques have safe memory access as a central goal. 
% 
% Accessing "foreign" memory
The Aida execution model~\cite{Lublinerman:2011:DI:2048066.2048133} uses delegation to allow a task to access memory associated another task. The tasks merge, becoming one task that has access to the memory of both and also the obligation to achieve the goals of both tasks. %This strategy favors parallel algorithms and the paper shows how it can be applied to irregular shared-memory applications. However, it does not, nor does it have to, address the situation of having persistent processes that exchange resources with each other: such processes would be merged.
%
%Connection of synchronization measures to data
Vaziri et al.~\cite{Vaziri:2006:ASC:1111037.1111067} introduce a method for specifying atomic sets and a static analysis that identifies when locks have to be taken and released. It relieves the programmer from defining locks and correctly using them. However, this automatism also limits the expressiveness and does not take into account other synchronization mechanisms such as channels or queues and does not allow for atomic sets to change dynamically. Our approach is agnostic to the synchronization measures used and allows for flexible aggregation.
%
%Avoiding of low level data races by design
%Several programming models have the avoidance of data races as a goal. 
The original monitor concept~\cite{Hoare:1974:MOS:355620.361161}, where a monitor protects all its associated shared resources (its attributes), follows the idea to make the connection between locks and resources explicit. Message-passing programming models, such as the Actor model~\cite{Hewitt:1973:UMA:1624775.1624804} or SCOOP~\cite{nienaltowski:2007:SCOOP}, simply avoid incorrect memory access by not sharing state; in data-intensive parallel applications, this approach incurs much data messaging overhead, which might be prohibitive unless combined with other techniques~\cite{schill-nanz-meyer:2013:handling}. Our approach applies to shared-memory environments. Approaches for safe parallel programming such as Deterministic Parallel Java~\cite{Bocchino:2009} ensures deterministic execution of parallel programs, thereby also avoiding race conditions, but requires a significant annotation overhead.

Various solutions have been proposed to detect low-level data races. Among them are tools to detect such races statically~\cite{Engler:2003:RES:945445.945468,Naik:2006:ESR:1133981.1134018,Boyapati:2002:OTS:582419.582440,Flanagan:2000:TRD:349299.349328,Anderson:2009:LAC:1542476.1542488}. As static techniques rely on an approximate information, they necessarily report false positives, and excessive reporting of race warnings can limit the usefulness of the approach. Dynamic tools based on locksets~\cite{savage1997eraser, vonPraun:2001:ORD:504282.504288} or happens-before relations~\cite{O'Callahan:2003:HDD:781498.781528,Marino:2009:LES:1542476.1542491,Flanagan:2009:FEP:1542476.1542490,Bond:2010:PPD:1806596.1806626} are typically more precise, but not sound, since not all execution paths are considered. SOM works on the higher abstraction level of checking for ownership violations and is therefore affected less by different execution paths, while also detecting more high-level data races.

\section{Conclusion}
\label{sec:conclusion}

Safe resource access is one of the fundamental concerns in shared-memory concurrent programming. A wide variety of synchronization mechanisms are available to the programmer, but their usage is prone to programming mistakes. Automatically checking correct usage is also difficult because of the missing explicit connection between the mechanisms and the resources they are supposed to protect. 

In this paper, we approached this issue by introducing a generalization of the ownership methodology, shared ownership, which uses a graph-based formalism to express the intended ownership of processes over resources. We demonstrated that shared ownership is compatible with a wide variety of synchronization mechanisms, in particular also readers-writer locks, which cannot be handled by classic ownership approaches. On the basis of this methodology, we then defined an abstract model, SOM, for tracking ownership changes, and applied it in a dynamic technique for checking ownership violations in Java programs. We demonstrated that the technique typically incurs only a moderate overhead; in addition, the model computation can be cleanly removed for production versions, with negligible overhead remaining.

SOM is a general framework for reasoning about shared ownership: it abstracts away from the actual synchronization. Using the dynamic simulation technique presented in this paper is only one way to apply this framework. In future work, we want to develop static techniques to check for ownership violations on the basis of SOM, as well as use it together with general program verification techniques, such as axiomatic semantics. Furthermore, shared ownership could be used for improving the specification of classes and methods by integrating checks into contracts, such as pre-, postconditions and class invariants. %In order to explore the interaction of SOM with contracts, an integration into a language with support for contracts such as Eiffel is needed. Contracts can improve the usefulness of SOM checks because a precondition implies that all the resources involved are owned by the thread invoking the routine. Adding contracts to channels, locks and the pass statement could be used to check for adherence to communication protocols. 

\paragraph{Acknowledgments.} This work has been supported in part by ERC Grant no. 291389.
%The research leading to these results has received funding from the European Research Council under the European Union's Seventh Framework Programme (FP7/2007-2013) / ERC Grant agreement no. 291389.

\clearpage

\bibliographystyle{splncs03}
\bibliography{bibfile}

\end{document}

%% file: pp.table.tex
% latex table generated in R 3.0.1 by xtable 1.7-1 package
% Thu Oct 10 15:30:43 2013
\begin{tabular}{@{}rrrrrrrrr}
  \hline
 & base & SD & full & SD & part. & SD & none & SD \\ 
  \hline
10k & 75 & 9 & 111 & 3 & 102 & 3 & 80 & 10 \\ 
  20k & 138 & 12 & 175 & 6 & 166 & 4 & 142 & 11 \\ 
  30k & 204 & 12 & 240 & 8 & 230 & 7 & 208 & 12 \\ 
  40k & 267 & 12 & 306 & 8 & 295 & 8 & 274 & 15 \\ 
  50k & 332 & 14 & 369 & 12 & 358 & 11 & 337 & 14 \\ 
  60k & 394 & 15 & 433 & 14 & 424 & 10 & 401 & 14 \\ 
  70k & 460 & 16 & 497 & 15 & 489 & 13 & 467 & 15 \\ 
  80k & 524 & 16 & 564 & 15 & 553 & 14 & 534 & 20 \\ 
  90k & 588 & 17 & 624 & 19 & 616 & 16 & 594 & 21 \\ 
  100k & 655 & 23 & 691 & 20 & 682 & 17 & 661 & 20 \\ 
   \hline
\end{tabular}

%% file: sort.table.tex
% latex table generated in R 3.0.1 by xtable 1.7-1 package
% Thu Oct 10 15:30:43 2013
\begin{tabular}{@{}rrrrrrrrr}
  \hline
 & base & SD & full & SD & part. & SD & none & SD \\ 
  \hline
100k & 49 & 5 & 101 & 18 & 54 & 1 & 71 & 24 \\ 
  200k & 59 & 1 & 120 & 2 & 66 & 1 & 63 & 1 \\ 
  300k & 99 & 26 & 173 & 8 & 102 & 15 & 92 & 8 \\ 
  400k & 105 & 10 & 209 & 6 & 115 & 1 & 112 & 7 \\ 
  500k & 128 & 5 & 275 & 7 & 144 & 1 & 134 & 2 \\ 
  600k & 149 & 2 & 314 & 4 & 174 & 2 & 163 & 12 \\ 
  700k & 180 & 2 & 392 & 7 & 213 & 16 & 193 & 2 \\ 
  800k & 198 & 8 & 430 & 9 & 232 & 5 & 213 & 8 \\ 
  900k & 225 & 4 & 490 & 11 & 261 & 5 & 250 & 24 \\ 
  1000k & 258 & 4 & 570 & 13 & 291 & 5 & 271 & 4 \\ 
   \hline
\end{tabular}

%% file: delref.table.tex
% latex table generated in R 3.1.0 by xtable 1.7-3 package
% Fri Apr 25 16:07:53 2014
\begin{tabular}{rrrrrrr}
  \hline
 & full & SD & part. & SD & none & SD \\ 
  \hline
 1 & 124.2 & 9.2 & 91.3 & 1.3 & 77.7 & 1.5 \\ 
   2 & 77.7 & 5.1 & 62.1 & 3.8 & 54.8 & 2.3 \\ 
   3 & 55.2 & 0.8 & 49.7 & 3.3 & 42.7 & 2.2 \\ 
   4 & 47.3 & 1.7 & 40.5 & 2.3 & 36.6 & 0.9 \\ 
   5 & 42.8 & 2.3 & 36.9 & 2.7 & 32.7 & 1.9 \\ 
   6 & 39.0 & 3.4 & 33.2 & 1.5 & 30.6 & 1.4 \\ 
   7 & 38.1 & 1.3 & 32.2 & 1.8 & 30.0 & 1.2 \\ 
   8 & 35.6 & 1.7 & 31.0 & 1.5 & 28.7 & 1.3 \\ 
   \hline
\end{tabular}